\documentclass[11pt,a4paper]{amsart}

\usepackage{amsopn}
\usepackage{amsmath}
\usepackage{amssymb}
\usepackage{amsthm}
\usepackage[foot]{amsaddr}
\usepackage{hyperref}
\usepackage{graphicx}
\usepackage{color}
\usepackage{listings}
\usepackage{subcaption}
\usepackage{enumerate}
\usepackage{cite}
\usepackage{todonotes}
\usepackage{mathtools}

\usepackage{multirow}
\usepackage{adjustbox}
\usepackage{geometry}

\usepackage{tikz}
\usetikzlibrary{angles,quotes}
\usetikzlibrary{decorations.pathreplacing}
\newcommand{\ie}{{\emph{i.e.\/}}}

\DeclareMathOperator{\tr}{tr}
\DeclareMathOperator{\Tr}{Tr}

\newcommand{\ket}[1]{\ensuremath{|#1\rangle}}
\newcommand{\bra}[1]{\ensuremath{\langle#1|}}
\newcommand{\ketbra}[2]{\ensuremath{\ket{#1} \! \bra{#2}}}
\newcommand{\proj}[1]{\ensuremath{\ketbra{#1}{#1}}}
\newcommand{\braket}[2]{\ensuremath{\langle{#1}|{#2}\rangle}}

\newcommand{\1}{{\rm 1\hspace{-0.9mm}l}}

\newcommand{\SPAN}{\mathrm{span}}

\newcommand{\PP}{\mathcal{P}}

\newcommand{\ketV}[1]{\ensuremath{|#1\rangle\!\rangle}}
\newcommand{\braV}[1]{\ensuremath{\langle\!\langle#1|}}
\newcommand{\ketbraV}[2]{\ensuremath{\ketV{#1}\braV{#2}}}
\newcommand{\projV}[1]{\ensuremath{\ketbraV{#1}{#1}}}

\newcommand{\supp}{\mathrm{supp}}

\theoremstyle{definition}

\newtheorem{lemma}{Lemma}
\newtheorem{theorem}{Theorem}

\newtheorem{corollary}{Corollary}
\newtheorem{proposition}{Proposition}

\newtheorem{remark}{Remark}
\newtheorem*{rem*}{Remark}

\renewcommand{\AA}{\mathbb{A}}
\renewcommand{\P}{\mathbb{P}}

\usepackage{caption}
\def\>{\rangle}
\def\<{\langle}
\title{Excluding false negative error in certification of quantum channels}
\date{June 4, 2021}

\author{Aleksandra Krawiec$^{1,*}$ \and
{\L}ukasz Pawela$^1$ \and 
Zbigniew Pucha{\l}a$^{1,2}$}
\address{$^1$ 
Institute of Theoretical and Applied Informatics, Polish Academy of Sciences, 
ul. Ba{\l}tycka 5, 44-100 Gliwice, Poland}
\address{$^2$ Faculty of Physics, Astronomy and Applied Computer Science, 
Jagiellonian University, ul. {\L}ojasiewicza 11, 30-348 Krak{\'o}w, Poland}
\email{akrawiec@iitis.pl}

\begin{document}
\maketitle

\begin{abstract}
Certification of quantum channels 
is based on quantum hypothesis testing and involves also preparation of an input
state and choosing the final measurement. 
This work primarily focuses on the
scenario when the false negative error cannot occur, even if it leads to the
growth of the probability of false positive error. We establish a condition when
it is possible to exclude false negative error after a finite number of queries
to the quantum channel in parallel, and we provide an upper bound on the number
of queries. 
On top of that, we found a class of channels which allow for excluding false
negative error after a finite number of queries in parallel, but cannot be
distinguished unambiguously.  
Moreover, it will be proved
that parallel certification scheme is always sufficient,
however the number of steps may be decreased by the use of adaptive scheme.
Finally, we consider examples of certification of various classes of quantum
channels and measurements.
\end{abstract}

\section{Introduction}
Being deceived is not a nice experience. People have been developing
plenty of methods to protect themselves against being cheated and one of
these methods concerns verification of objects, also quantum ones.
The cornerstone for theoretical studies on discrimination of quantum objects 
was laid by Helstrom~\cite{helstrom1976quantum} a few decades ago.

In the era of Noisy Intermediate-Scale Quantum (NISQ) devices
\cite{preskill2018quantum,bharti2021noisy}, assuring the correctness of
components in undeniably in the spotlight. A broad review of multipronged modern
methods of certification as well as benchmarking of quantum states and processes
can be found in the recent paper~\cite{eisert2020quantum}. For a more
introductory tutorial to the theory of system certification we refer the reader
to~\cite{kliesch2021theory}. 
Verification of quantum processes is often studied in the context of specific
elements of quantum information processing tasks. Protocols for efficient
certification of quantum processes, such as quantum gates and circuits, were
recently studied in~\cite{liu2020efficient,zeng2020quantum,zhu2020efficient}.

Let us introduce the most general problem of verification studied in this work.
Assume there are two known quantum channels and one of them is secretly chosen.
Then, we are given the secretly chosen channel to verify which of the two
channels it is. We are allowed to prepare any input state and apply the given
channel on it. Finally, we can prepare any quantum measurement and measure the
output state. Basing on the measurement's outcome we make a decision which of
the two channels was secretly chosen. 
In this work we focus on the case when we are promised which of the channels 
is given. After performing some certification procedure we can either agree 
that the channel was the promised one or claim that we were cheated.
We want to assure that we will always realize when we are cheated. It may 
happen though, that we appear to be too suspicious and claim that we were 
cheated when we were not.
 
There are three major theoretical
approaches towards verification of quantum channels called minimum error
discrimination, unambiguous discrimination and certification. All these three
approaches can be generalized to the multiple-shot case, that is when the given
channel can be used multiple times in various configurations. The most
straightforward possibility is the parallel scheme
and the most sophisticated is the adaptive scheme (where we are allowed
to use any processing between the uses of the given channel).

The first approach is called minimum error discrimination (a.k.a. 
distinguishability or symmetric discrimination) and makes use of the distance 
between quantum channels expressed by the use of the diamond norm. 
In this scenario one wants to minimize the probability of making the erroneous 
decision using the bound on this probability given by the Holevo-Helstrom 
theorem \cite{helstrom1976quantum,watrous2018theory}.
Single-shot discrimination of unitary channels and von Neumann measurements 
were studied in~\cite{ziman2010single,bae2015discrimination} and 
\cite{ji2006identification,sedlak2014optimal,puchala2018strategies} 
respectively. 
Parallel discrimination of quantum channels was 
studied eg. in \cite{duan2016parallel,cao2016minimal}. It appeared that 
parallel discrimination 
scheme is optimal in the case of distinguishability of unitary 
channels~\cite{chiribella2008memory} and von Neumann 
measurements~\cite{puchala2021multiple}. In some cases however, 
the use of adaptive discrimination scheme can significantly improve the 
certification~\cite{harrow2010adaptive,krawiec2020discrimination}.
Advantages of the use of adaptive discrimination scheme in the asymptotic 
regime were studied in~\cite{salek2020adaptive}.
Fundamental and ultimate limits for quantum channel discrimination were 
derived in \cite{pirandola2019fundamental,zhuang2020ultimate}.
The works \cite{katariya2020evaluating,wang2019resource} address the problem of 
distinguishability of quantum channels in the context of resource theory.

In the second approach, that is unambiguous discrimination, there are three
possible outcomes. Two of them designate
quantum channels while the third option is the inconclusive result.   
In this approach, when the result indicated which channel was given, we know it
for sure. There is a chance however, that we will obtain an inconclusive
answer. Unambiguous discrimination of quantum channels was considered  
in~\cite{wang2006unambiguous}, while unambiguous discrimination of von Neumann
measurements was explored in~\cite{puchala2021multiple}. Studies on unambiguous
discrimination of quantum channels took a great advantage of unambiguous
discrimination of quantum states, which can be found eg.
in~\cite{feng2004unambiguous,zhang2006unambiguous,
herzog2005optimum,bergou2006optimal,herzog2009discrimination}.

The third approach, known as certification or asymmetric discrimination, is
based on hypothesis testing. We are promised to be given one of the two channels
and associate this channel with the null hypothesis, $H_0$. The other channel is
associated with the alternative hypothesis, $H_1$. When making a decision 
whether to accept or to reject the null hypothesis, two types of errors may 
occur, that is we can come across false positive and false negative errors.  
In this work we consider the situation when we want to assure that false
negative error will not occur, even if the probability of false positive error
grows. A similar task of minimizing probability of false negative error having
fixed bound on the probability of false positive was studied in the case of von
Neumann measurements in~\cite{lewandowska2021optimal}. Certification of quantum
channels was studied in the asymptotic regime e.g. 
in~\cite{wilde2020amortized,salek2020adaptive,katariya2020evaluating}.

It should come as no surprise that in some cases perfect verification is not
possible by any finite number of steps. Conditions for perfect minimum error
discrimination of quantum operations were derived in~\cite{duan2009perfect}.
Similar condition for unambiguous discrimination was proved
in~\cite{wang2006unambiguous}. However, no such conditions have been stated for
certification. In this work we derive a condition when we can exclude false
negative error after a finite number of uses in parallel. This condition holds
for arbitrary quantum channels and is expressed by the use of Kraus operators of
these channels. We will provide an example of channels which can be certified in
a finite number of queries in parallel, but cannot be distinguished
unambiguously. Moreover, we will show that, in contrast to discrimination of
quantum channels \cite{harrow2010adaptive,krawiec2020discrimination},  parallel
certification scheme is always sufficient for certification, although the number
of uses of the certified channel may not be optimal. On top of that, we will
consider certification of quantum measurements and focus on the class of
measurements with rank-one effects. The detailed derivation of the upper bound
for the probability of false positive error will be presented for SIC POVMs.

This work is organized as follows. After preliminaries in
Section~\ref{sec:preliminaries}, we present our main result, that is  the
condition when excluding false negative is possible in a finite number of uses
in parallel, in Theorem~\ref{th:condition_for_parallel_certification} in
Section~\ref{sec:condition_parallel_certification}. 
Certification of quantum
measurements is discussed in Section~\ref{sec:certification_of_povms}. 
Then, in
Section~\ref{sec:adaptive_certification}  we study adaptive certification and
Stein setting. The condition when excluding false negative is possible in
adaptive scheme is stated therein as
Theorem~\ref{th:equivalence_parallel_adaptive}. 
Finally, conclusions can be found in Section~\ref{sec:conclusions}.

\section{Preliminaries}\label{sec:preliminaries}
Let $\mathcal{D}_d$ denote the set of quantum states of dimension $d$, that is  
the set of positive semidefinite operators having trace equal one. Throughout
this paper quantum states will be denoted by lower-case Greek letters, usually
$\rho, \sigma, \tau$. For any state $\rho \in \mathcal{D}_d$ we can write its
spectral decomposition as $\rho = \sum_i p_i \proj{\lambda_i}$. 
Having a set of quantum states $\{\rho_1, \ldots, \rho_m\}$ with spectral 
decompositions $\rho_1 = \sum_{i_1} p_{i_1} \proj{\lambda_{i_1}}, \ldots, 
\rho_m = \sum_{i_m} p_{i_m} \proj{\lambda_{i_m}}$ respectively, their 
\emph{support} is defined as 
$\supp(\rho_1, \ldots, \rho_m) \coloneqq \SPAN \{ \ket{\lambda_{i_j}}: p_{i_j} 
> 0 \}$.
The set of unitary matrices of dimension $d$ will be denoted $\mathcal{U}_d$.

Quantum channels are linear maps which are completely positive and trace
preserving. In this work we will often take advantage of the Kraus
representations of channels. Let
\begin{equation}
\Phi_0(X) := \sum_{i=1}^k E_i X E_i^\dagger,
\quad 
\Phi_1(X) := \sum_{j=1}^l F_j X F_j^\dagger
\end{equation}
be the Kraus representations of the channels that will correspond to null and
alternative hypotheses, respectively. The sets of operators $\{E_i\}_i$ and
$\{F_j\}_j$ are called Kraus operators of channels $\Phi_0$ and $\Phi_1$,
respectively. We will use the notation $\supp(\Phi_0) \coloneqq \SPAN\{E_i
\}_i$, $\supp(\Phi_1) \coloneqq \SPAN\{F_j \}_j$, to denote the supports of
quantum channels.
Moreover, the notation $\1$ will be used for the identity channel.

The most general quantum measurements, known also as POVMs (positive operator
valued measure) are defined as a collection of positive semidefinite operators
$\PP = \{ M_1, \ldots, M_m \}$ which fulfills the condition $\sum_{i=1}^m M_i =
\1_d$, where $\1_d$ denotes the identity matrix of dimension $d$. When a quantum
state $\rho$ is measured by the measurement $\PP$, then the label $i$ is
obtained with probability $\Tr (E_i \rho)$ and the state $\rho$ ceases to exist.
A special class of quantum measurements are projective von Neumann measurements.
These POVMs have rank-one effects of the form $\{ \proj{u_1}, \ldots, \proj{u_d}
\}$, where vectors $\{\ket{u_i}\}_{i=1}^d$ form an orthonormal basis and 
therefore they are columns of some unitary matrix $U \in\mathcal{U}_d$.

Now we proceed to describing the detailed scheme of certification. There are two
quantum channels: $\Phi_0$ and $\Phi_1$. We are promised that we are given
$\Phi_0$ but we are not sure and we want to verify it using hypothesis testing.
We associate the channel $\Phi_0$ with the null hypothesis $H_0$ and we
associate the other channel $\Phi_1$ with the alternative hypothesis $H_1$. We
consider the following scheme. We are allowed to prepare any (possibly
entangled) input state and perform the given channel on it. Then, we prepare a
binary measurement $\{\Omega_0, \1- \Omega_0\}$ and measure the output state. If
we obtain the label associated with the effect $\Omega_0$, then we decide that
the certified channel was $\Phi_0$ and we accept the null hypothesis. If we get
the label associated with the effect   $\1- \Omega_0$, then we decide that the
certified channel was $\Phi_1$ and therefore we reject the null hypothesis.

The aim of certification is to make a decision whether to accept or to reject
$H_0$. While making such a decision one can come upon two types of errors. The
false positive error (also known as type I error) happens when we reject the
null hypothesis when in fact it was true. The converse situation, that is
accepting the null hypothesis when the alternative hypothesis was correct, is
known as the false negative (or type II) error. In this work we will focus on
the situation when the probability of the false negative error equals zero and
we want to minimize the probability of false positive error.

Let us now take a closer look into the scheme of entanglement-assisted
single-shot certification procedure. We begin with preparing an input state
$\ket{\psi}$ on the compound space. Then, we apply the certified channel
extended by the identity channel on the input state, obtaining as the output the
state either $\rho_0^{\ket{\psi}} = \left( \Phi_0 \otimes \1 \right)
(\proj{\psi})$, if the given channel was $\Phi_0$, or $\rho_1^{\ket{\psi}} =
\left( \Phi_1 \otimes \1 \right) (\proj{\psi})$, if the given channel was
$\Phi_1$. Eventually, we perform the measurement $\{\Omega_0, \1 - \Omega_0 \}$,
where the effect $\Omega_0$ accepts hypothesis $H_0$ and the effect
$\1-\Omega_0$ accepts the alternative hypothesis $H_1$.

Assuming that the input state $\ket{\psi}$ and measurement effect $\Omega_0$
have been fixed, the probability of making the false
positive error is given by
\begin{equation}\label{def_of_p1_conditional_single_shot}
p_1 \left(\ket{\psi}, \Omega_0 \right)
\coloneqq \Tr \left( (\1 - \Omega_0) 
\rho_0^{\ket{\psi}} \right) 
= 1- \Tr \left( \Omega_0 \rho_0^{\ket{\psi}} \right).
\end{equation}
In a similar manner we have the probability of making the false
negative error, that is
\begin{equation}
p_2 \left( \ket{\psi}, \Omega_0 \right) 
\coloneqq \Tr \left( \Omega_0 \rho_1^{\ket{\psi}} \right).
\end{equation}

We will be interested in the situation when probability of  the false
negative error is equal to zero and we want to minimize the probability of false
positive error. Therefore, we introduce the notation
\begin{equation}\label{eq:def_of_p1}
p_1 \coloneqq \min_{\ket{\psi}, \Omega_0} 
\left\{p_1 \left(\ket{\psi}, \Omega_0 \right): \ 
p_2 \left(\ket{\psi}, \Omega_0 \right) = 0
\right\}
\end{equation}
for minimized probability of false positive error in the single-shot scenario.

For a given $\epsilon >0$, we say that quantum channel $\Phi_0$ \emph{can be 
$\epsilon$-certified against} channel
$\Phi_1$
if there exist an input state
$\ket{\psi}$ and measurement effect $\Omega_0$ such that   $p_2
\left(\ket{\psi}, \Omega_0 \right) = 0$ and $p_1 \left(\ket{\psi}, \Omega_0
\right) \leq \epsilon$.
In other words, quantum channel $\Phi_0$ can be $\epsilon$-certified against 
another channel $\Phi_1$ if we can assure no false negative will occur and the
probability of false positive error is smaller than $\epsilon$.

When performing the certification of quantum channels, we can use the channels
many times in various configurations. Now we proceed to introducing notation
needed for studying parallel and adaptive certification schemes.

\subsection{Parallel certification scheme}
Let $N$ denote the number of uses of the quantum channel in parallel. A
schematic representation of the scenario of parallel certification is depicted
in Figure~\ref{fig:parallal_scheme}. In this scheme we consider certifying 
tensor products of the channels.
In other words, parallel certification of channels $\Phi_0$ and $\Phi_1$ can be
seen as certifying channels $\Phi_0^{\otimes N}$ and $\Phi_1^{\otimes N}$ for
some natural number $N$. 
\begin{figure}[h!]
\includegraphics[scale=0.45]{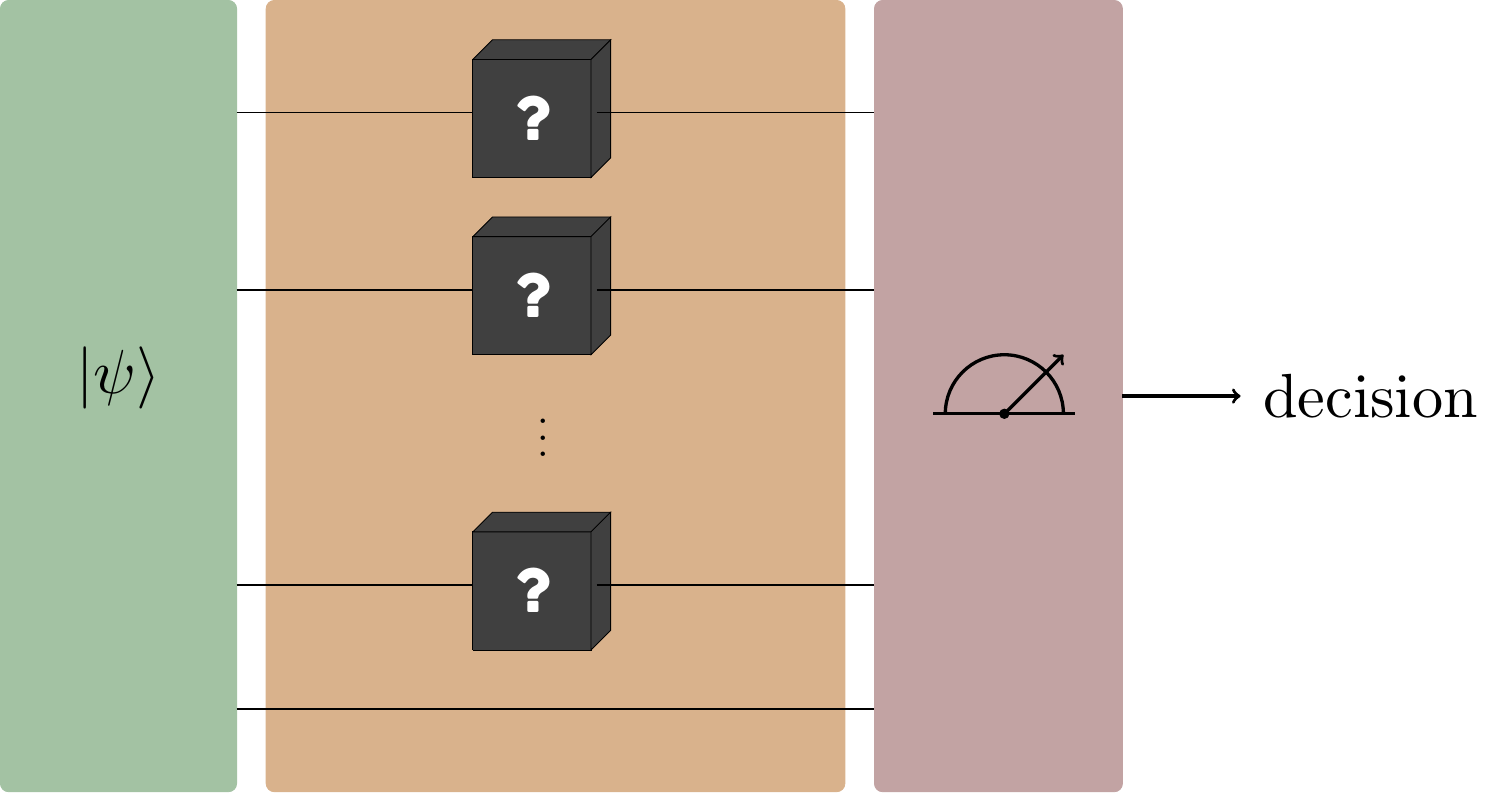}
\caption{Parallel certification scheme}
\label{fig:parallal_scheme}
\end{figure}

Let $\ket{\psi}$ be the input state to the certification procedure. After
applying the channel $\Phi_0$ $N$ times in parallel, we obtain the output state
\begin{equation}
\sigma_0^{N, \ket{\psi}} = \left(\Phi_0^{\otimes N} \otimes \1\right) 
(\proj\psi),
\end{equation}
if the channel was $\Phi_0$, and similarly 
\begin{equation}
\sigma_1^{N, \ket{\psi}} = \left(\Phi_1^{\otimes N} \otimes \1\right) 
(\proj\psi),
\end{equation}
if the channel was $\Phi_1$.
In the same spirit let 
\begin{equation}\label{eq:def_of_errors_in_parallel_scheme}
p_1^{\P, N} \left(\ket{\psi}, \Omega_0 \right) = \Tr \left((\1 - \Omega_0) 
\sigma_0^{N, \ket{\psi}}\right) , 
\quad p_2^{\P, N} \left(\ket{\psi}, \Omega_0 \right) = \Tr\left(\Omega_0 
\sigma_1^{N, \ket{\psi}} \right)
\end{equation}
be the probabilities of occurring false positive and false negative errors 
respectively.
When $N=1$, then we arrive at single-shot certification. Therefore we will
neglect the upper index and simply write 
$p_1 \left(\ket{\psi}, \Omega_0 \right)$ and 
$p_2 \left(\ket{\psi}, \Omega_0 \right)$.

We introduce the notation
\begin{equation}\label{eq:def_of_p1_parallel}
p_1^{\P, N} \coloneqq \min_{\ket{\psi}, \Omega_0} 
\left\{p_1^{\P, N} \left(\ket{\psi}, \Omega_0 \right): \ 
p_2^{\P, N} \left(\ket{\psi}, \Omega_0 \right) = 0
\right\}
\end{equation}
for the minimized probability of false positive error in the parallel scheme.

We say that quantum channel $\Phi_0$ \emph{can be certified against} $\Phi_1$
\emph{in the parallel scheme},
if for every $\epsilon >0$ there
exist a natural number $N$, an input state $\ket{\psi}$ and measurement effect
$\Omega_0$ such that   $p_2^{\P, N} \left(\ket{\psi}, \Omega_0 \right) = 0$ and
$p_1^{\P, N} \left(\ket{\psi}, \Omega_0\right) \leq \epsilon$.

Let us now elaborate a bit on the number of steps needed for certification. 
Assume that we have fixed upper bound on the probability of false positive 
error, $\epsilon >0$. We will be interested in calculating the minimal number 
of queries, $N_\epsilon$, for which $p_2^{\P, N} \left(\ket{\psi}, \Omega_0 \right) = 0$ 
and $p_1^{\P, N} \left(\ket{\psi}, \Omega_0\right) \leq \epsilon$ for some 
input state $\ket{\psi}$ and measurement effect $\Omega_0$. Such a number, 
$N_\epsilon$, will be called the \emph{minimal number of steps needed for 
parallel certification}.

\subsection{Adaptive certification scheme}
Adaptive certification scheme allows for the use of processing between the uses 
of the certified channel, therefore this procedure is more complex then the 
parallel certification. However, when the processings only swap the subsystems, 
then the adaptive scheme may reduce to the parallel one.

\begin{figure}[h!]
\includegraphics[scale=0.45]{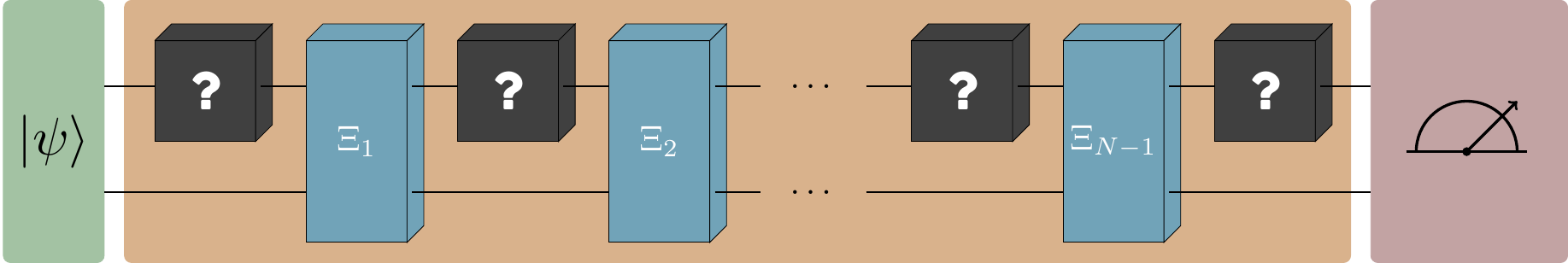}
\caption{Adaptive certification scheme. The processings $\Xi_1, \ldots, 
\Xi_{N-1}$ can be arbitrary quantum channels where we only assume that the 
first subsystem must fit the input of the black box. In particular, the 
processings can swap the subsystems, and therefore one can obtain the parallel 
scheme as a special case of adaptive discrimination scheme.}
\label{fig:adaptive_scheme}
\end{figure} 
Assume as previously that $\ket{\psi}$ is the input state to the certification 
procedure in which the certified channel in used $N$ times and any processing 
is allowed between the uses of this channel. 
The scheme of this procedure is presented in the 
Figure~\ref{fig:adaptive_scheme}.
Having the input state $\ket\psi$ on the compound register, we perform the
certified channel (denoted by the black box with question mark) on one part of
it. Having the output state we can perform some processing $\Xi_1$ and therefore
get prepared for the next use of the certified channel. Than again, we apply the
certified channel on one register of the prepared state and again, we can
perform processing $\Xi_2$. We repeat this procedure $N-1$ times. After the
$N$-th use of the certified channel we obtain the state either $\tau_0^{N,
\ket{\psi}}$, if the channel was $\Phi_0$, or $\tau_1^{N, \ket{\psi}}$, if the
channel was $\Phi_1$. Then,  we prepare a global measurement $\{\Omega_0, \1 -
\Omega_0\}$ and apply it on the output state. 
Let
\begin{equation}
p_1^{\AA, N} \left(\ket{\psi}, \Omega_0 \right) = \Tr \left((\1 - \Omega_0) 
\tau_0^{N, \ket{\psi}}\right) , 
\quad p_2^{\AA, N} \left(\ket{\psi}, \Omega_0 \right) = \Tr\left(\Omega_0 
\tau_1^{N, \ket{\psi}} \right)
\end{equation}
be the probabilities of the false positive and false negative errors in adaptive
scheme, respectively, when the input state and the measurement effects
were fixed. When $N=1$, then we will neglect the upper index and simply write
$p_1 \left(\ket{\psi}, \Omega_0 \right)$ and $p_2 \left(\ket{\psi}, \Omega_0
\right)$.

We say that quantum channel $\Phi_0$ \emph{can be certified against} $\Phi_1$
\emph{in the adaptive scheme},
if for every $\epsilon >0$ there
exist a natural number $N$, an input state $\ket{\psi}$ and measurement effect
$\Omega_0$ such that   $p_2^{\AA, N} \left(\ket{\psi}, \Omega_0 \right) = 0$ and
$p_1^{\AA, N} \left(\ket{\psi}, \Omega_0\right) \leq \epsilon$.

For a fixed upper bound on the probability of false positive error, $\epsilon$, 
we introduce the \emph{minimal number of steps needed for adaptive 
certification}, $N_\epsilon$, as the minimal number of steps after which  
$p_2^{\AA, N} \left(\ket{\psi}, \Omega_0 \right) = 0$ 
and $p_1^{\AA, N} \left(\ket{\psi}, \Omega_0\right) \leq \epsilon$ for some 
input state $\ket{\psi}$ and measurement effect $\Omega_0$. 

\section{Parallel certification}\label{sec:condition_parallel_certification}

Not all quantum channels can be discriminated perfectly after a finite number of
queries. Conditions for perfect discrimination were states in the
work~\cite{duan2009perfect}. Similar conditions for unambiguous discrimination
were proved in~\cite{wang2006unambiguous}. In this section we will complement
these results with the condition concerning parallel certification. 
More specifically, we will prove a simple necessary and sufficient condition 
when a quantum channel $\Phi_0$ can be certified against 
some other channel $\Phi_1$. 
As the condition utilizes the notion of the support of a quantum channel, 
recall that it is defined as the span of their Kraus operators.  
The condition will be stated as
Theorem~\ref{th:condition_for_parallel_certification}, however its proof will be
presented after introducing two technical lemmas.

In fact, the statement of Theorem~\ref{th:condition_for_parallel_certification} 
is a bit more general, that is it concerns the situation when the alternative 
hypothesis corresponds to  a set of channels $ \left\{ \Phi_1, 
\ldots , \Phi_m \right\}$ having 
Kraus operators 
$\left\{ F^{(1)}_{{j_1}} \right\}_{{j_1}},
\ldots, 
\left\{ F^{(m)}_{{j_m}} \right\}_{{j_m}}$
respectively.
More precisely, the alternative hypothesis corresponds to the situation when 
every black box contains one of the channels 
$ \left\{ \Phi_1, \ldots , \Phi_m \right\}$ but not necessarily the same.
We will use the notation 
$\supp \left(\Phi_1, \ldots , \Phi_m \right) \coloneqq 
\SPAN \left\{ F^{(1)}_{{j_1}} , \ldots ,  F^{(m)}_{{j_m}} 
\right\}_{ {j_1}, \ldots,{j_m}}$.

\begin{theorem}\label{th:condition_for_parallel_certification}
Quantum channel $\Phi_0$ can be certified against quantum channels $\Phi_1, 
\ldots, \Phi_m$ in
the parallel scheme
if and only if $\supp(\Phi_0) \not\subseteq \supp \left(\Phi_1, \ldots, \Phi_m 
\right)$. 

Moreover, to ensure that the probability of false 
positive error is no greater than $\epsilon$, the number of steps needed for 
parallel certification is bounded by $N_\epsilon  \geq \left\lceil \frac{\log 
\epsilon}{\log p_1}\right\rceil$, 
where $p_1$ is the upper bound on probability of false positive error in 
single-shot certification.
\end{theorem}

Before presenting the proof of this theorem we will introduce two lemmas. The
proofs of lemmas are postponed to Appendix~\ref{app:proofs_of_lemmas}.
Lemma~\ref{lm:inclusion_of_supports} states that if the inclusion does not hold
for supports of the quantum channels, then the inclusion also does not hold for
supports of output states assuming that the input state has full Schmidt rank.
The proof of Lemma~\ref{lm:inclusion_of_supports} is based on the proof 
in~\cite[Theorem~$1$]{wang2006unambiguous}, which studies unambiguous
discrimination among quantum operations.

\begin{lemma}\label{lm:inclusion_of_supports}
Let $\{\ket{a_t}\}_t$ and $\{\ket{b_t}\}_t$ be two orthonormal bases and 
$\ket{\psi} \coloneqq \sum_t \lambda_t \ket{a_t} \ket{b_t}$ where 
$\lambda_t>0$ for every $t$. 
Let also $\rho_0^{\ket{\psi}} = \left(\Phi_0 \otimes \1\right) (\proj\psi)$ and 
$\rho_j^{\ket{\psi}} = \left(\Phi_j \otimes \1\right) (\proj\psi)$ for $j=1, 
\ldots,  m$. 
If $\supp(\Phi_0) \not\subseteq \supp(\Phi_1, \ldots ,\Phi_m)$, then 
$\supp\left(\rho_0^{\ket{\psi}}\right) \not\subseteq 
\supp\left(\rho_1^{\ket{\psi}} , \ldots , \rho_m^{\ket{\psi}}\right)$.
\end{lemma}

Lemma~\ref{lm:inclusion_of_supports_2} also concerns inclusions of supports. It
states that if the inclusion of supports does not hold for some output states,
then it does not hold also for supports of the channels.

\begin{lemma}\label{lm:inclusion_of_supports_2}
With the notation as above, if there exists a natural number $N$ and an input 
state $\ket{\psi}$ such that 
$\supp\left( \sigma_0^{N, \ket{\psi}}
\right) \not\subseteq
\supp \left( \sigma_1^{N, \ket{\psi}}, \ldots, \sigma_m^{N, \ket{\psi}}
\right)$, 
then $\supp (\Phi_0) \not\subseteq \supp \left(\Phi_1, \ldots, \Phi_m\right)$.
\end{lemma}

Finally, we are in position to present the proof of
Theorem~\ref{th:condition_for_parallel_certification}.

\begin{proof}[Proof of Theorem~\ref{th:condition_for_parallel_certification}]
($\impliedby$) Let $\supp(\Phi_0) \not\subseteq \supp \left(\Phi_1, \ldots, 
\Phi_m \right)$.
From Lemma~\ref{lm:inclusion_of_supports} this implies 
$\supp\left(\rho_0^{\ket{\psi}} \right) 
\not\subseteq \supp\left(\rho_1^{\ket{\psi}} , \ldots , 
\rho_m^{\ket{\psi}}\right)$ where the input state is
$\ket{\psi}=\sum_t \lambda_t \ket{a_t} \ket{b_t}$. 
Hence we can always find a state $\ket{\phi_0}$ for which
\begin{equation}
\ket{\phi_0} \not\perp \supp\left(\rho_0^{\ket{\psi}} \right)  \quad 
\textrm{and} \quad 
\ket{\phi_0} \perp \supp\left(\rho_1^{\ket{\psi}} , \ldots , 
\rho_m^{\ket{\psi}}\right),
\end{equation}
and therefore
\begin{equation}
\bra{\phi_0}\rho_0^{\ket{\psi}} \ket{\phi_0} >0  \quad \textrm{and} \quad 
\bra{\phi_0}\rho_i^{\ket{\psi}} \ket{\phi_0} = 0
\end{equation}
for $i = 1 , \ldots, m$.

Now we consider the certification scheme by taking the measurement with effects
$\{\Omega_0, \1 - \Omega_0 \}$. Without loss of generality we can assume that
$\Omega_0 \coloneqq \proj{\phi_0}$ is a rank-one operator.  
We calculate
\begin{eqnarray}
\begin{split}
&\tr \left(\Omega_0 \rho_0^{\ket{\psi}} \right) = \bra{\phi_0} 
\rho_0^{\ket{\psi}} \ket{\phi_0} >0 \\
&p_2 \left( \ket{\psi}, \Omega_0 \right) 
= \sum_{i=1}^{m} \tr \left(\Omega_0 
\rho_i^{\ket{\psi}} \right) 
= \sum_{i=1}^{m} \bra{\phi_0} \rho_i^{\ket{\psi}} \ket{\phi_0} =0  \\
&p_1 \left( \ket{\psi}, \Omega_0 \right) = \tr \left((\1 - 
\Omega_0) \rho_0^{\ket{\psi}} \right) = 1- \bra{\phi_0} \rho_0^{\ket{\psi}} 
\ket{\phi_0}  < 1
\end{split}
\end{eqnarray}
Hence after sufficiently many uses, $N$, of the certified channel in parallel 
(actually when $N \geq \left\lceil \frac{\log \epsilon}{\log p_1}\right\rceil$)
we obtain that $\tr \left(\Omega_1^{\otimes N} 
{\left(\rho_0^{\ket{\psi}}\right)}^{\otimes N} \right) 
\leq 
\epsilon$ 
for any positive $\epsilon$. 
Therefore after $N$ queries we will be able to exclude false negative error.

($\implies$) Assume that $\Phi_0$ can be 
certified against $\Phi_1, \ldots , \Phi_m$ in the parallel 
scenario. This means that there exist a natural number $N$, an input state 
$\ket{\psi}$ and a positive operator (measurement effect) $\Omega_0$ on the 
composite system such that
\begin{equation}
\begin{split}
&p_1^{\P, N} \left( \ket{\psi}, \Omega_0\right) = 1- \tr \left(\Omega_0 \left( 
\Phi_0^{\otimes N} \otimes \1 
\right)( \proj{\psi})\right) \leq  \epsilon <1  \\
%&p_2^{\P, N} \left( \ket{\psi}, \Omega_0\right)
%= \sum_{i=1}^{m} \tr \left(\Omega_0 \left( 
%\Phi_i^{\otimes N} \otimes \1 \right)( 
%\proj{\psi})\right) =0. 
&\tr \left(\Omega_0 \left( 
\Phi_{i_1} \otimes \ldots \otimes  \Phi_{i_N} \otimes \1 \right)( 
\proj{\psi})\right) =0, \quad \forall_{i_1, \ldots , i_N \in \{1, \ldots,m\}}. 
\end{split}
\end{equation}
Therefore $\tr \left(\Omega_0 \left( \Phi_0^{\otimes N} \otimes \1 
\right)( \proj{\psi})\right) >0$ and thus
\begin{equation}
\begin{split}
\Omega_0 &\not\perp \supp \left(\left( \Phi_0^{\otimes N} \otimes \1 
\right)\left( 
\proj{\psi}\right)\right)  
= \SPAN \left\{ \left( E_{i_1} \otimes \ldots \otimes E_{i_N} \otimes \1 
\right)\ket{\psi}\right\}_{i_1, \ldots, i_N}
\\
\Omega_0 &\perp 
\SPAN \left\{ \left( K_{l}^{\otimes N} \otimes \1 
\right)\ket{\psi}\right\}_{l_1, \ldots, l_N},
\end{split}
\end{equation}
where $\{K_l\}_l = \left\{ F^{(1)}_{{j_1}} , \ldots ,  
F^{(m)}_{{j_m}} \right\}_{ {j_1}, \ldots,{j_m}} $ is the set of all Kraus 
operators of channels $\Phi_i$ from the alternative hypothesis.

Hence
\begin{equation}\label{eq:to_be_contradicted}
\SPAN \left\{ \left( E_{i_1} \otimes \ldots \otimes E_{i_N} \otimes \1 
\right)\ket{\psi}\right\}_{i_1, \ldots, i_N}
\not\subseteq
\SPAN \left\{ \left( K_{l_1} \otimes \ldots \otimes K_{l_N} \otimes \1 
\right)\ket{\psi}\right\}_{l_1, \ldots, l_N}.
\end{equation}
The reminder of the proof follows directly from 
Lemma~\ref{lm:inclusion_of_supports_2}.
\end{proof}

It is worth mentioning that in the above proof the measurement effect 
$\Omega_0$ is a rank-one projection operator. This is sufficient to prove 
that quantum channel $\Phi_0$ can be certified against $\Phi_1$  
in the parallel scheme, but this is, in most of the cases, not optimal. 

\begin{remark}
If we consider the hypothesis testing scenario when we have to decide whether 
the unknown channel is either $\Phi_0$ or some fixed $\Phi_i$, with $\Phi_i 
\in  \{ \Phi_1, \ldots, \Phi_m\}$ (\ie we have either $\Phi_0^{\otimes N}$ or 
$\Phi_i^{\otimes N}$),
%If we assume that the alternative hypothesis corresponds to the situation when 
%one of the channels $\Phi_1, \ldots ,\Phi_m$ is in all black boxes, then the 
then the quantum channel $\Phi_0$ can be certified against $\Phi_{1}, \ldots , 
\Phi_m$ 
in the parallel scheme if and only if
\begin{equation}
\supp \left(  \Phi_0 \right)
\not\subseteq \supp \left(  \Phi_j \right)
\end{equation}
for every $j \in \{1, \ldots, m\}$.
\end{remark}

This remark follows directly from considering certification with simple 
alternative hypotheses.

In the remainder of this section we will discuss two examples. The first 
example shows that if quantum channels can be certified in the parallel scheme, 
then it does not have to imply that they can be discriminated unambiguously.
We will provide an explicit example of mixed-unitary channels which fulfill the 
condition from Theorem~\ref{th:condition_for_parallel_certification}, and 
therefore can be certified in the parallel scheme, but cannot be discriminated 
unambiguously.
In the second example we will consider the situation when the channel 
associated with the $H_1$ hypothesis is the identity channel and derive an 
upper bound on the probability of false positive error.

\subsection{Channels which cannot be discriminated unambiguously but still can 
be certified.}
In this subsection we will give an example of a class of channels which cannot 
be discriminated unambiguously, 
but they can be certified by a finite number of uses in the parallel scheme. 
The work~\cite{wang2006unambiguous} presents the condition when quantum 
channels can be unambiguously discriminated by a finite number of uses. More 
precisely, Theorem~2 therein states that if a set of quantum channels 
$\mathcal{S} = \{\Phi_i\}_i$ satisfies the condition  $\supp (\Phi_i) 
\not\subseteq \supp(\Phi_j)$ for every $\Phi_i, \Phi_j \in \mathcal{S}$, then 
they can be discriminated unambiguously in a finite number of uses.

Now we proceed to presenting our example.
Let $\Phi_0$ be a mixed unitary channel of the form
\begin{equation}
\Phi_0(\rho) = \sum_{i=1}^m p_i U_i \rho U_i^\dagger,
\end{equation}
where $p = (p_1, \ldots, p_m)$ is a probability vector and 
$\{  U_1, \ldots, U_m\}$ are unitary matrices.  
As the second channel we take a unitary channel of the form 
$\Phi_1(\rho) = \tilde{U} \rho \tilde{U}^\dagger$, where we make a crucial 
assumption that $\tilde{U} \in \{  U_1, \ldots, U_m\}$.

Therefore we have $\supp(\Phi_0) = \SPAN \{\sqrt{p_i}U_i \}_i$, 
while $\supp(\Phi_1) = \SPAN \{\tilde{U}\}$.
In this example it can be easily seen that the condition for unambiguous 
discrimination is not fulfilled as 
$\supp (\Phi_1) \subseteq \supp(\Phi_0)$. 
Nevertheless, the condition from 
Theorem~\ref{th:condition_for_parallel_certification} is fulfilled as 
$\supp (\Phi_0) \not\subseteq \supp(\Phi_1)$, and hence it is possible to 
exclude false negative error after a finite number of queries in parallel.

\subsection{Certification of arbitrary channel against the identity channel}
Assume that we want to certify channel $\Phi_0$, which Kraus operators are 
$\{E_i\}_i$, against the identity channel $\Phi_1$ having Kraus operator 
$\{\1\}$. We will show that as long as the channel $\Phi_0$ is not the 
identity channel, it can always be certified against the identity channel in 
the parallel scheme.

\begin{proposition}
Every quantum channel (except the identity channel) can be certified 
against the identity channel in the parallel scheme.
\end{proposition}

\begin{proof}
Let $\ket{\psi}$ be an input state. After applying the certified channels on 
it, we obtain the state either 
$\rho_0^{\ket{\psi}} = \left(\Phi_0 \otimes \1\right)\left(\proj{\psi}\right)$, 
if the 
channel was $\Phi_0$, or 
$\rho_1^{\ket{\psi}} = \proj{\psi}$, if the channels was $\Phi_1$. 
As the final measurement effect we can take 
$\Omega_0 \coloneqq \1 - \proj{\psi}$, which is always orthogonal to 
$\rho_1^{\ket{\psi}}$, hence no false negative error will occur.
Having the input state and final measurement fixed, we will calculate the 
probability of false positive error in the single-shot scheme
\begin{equation}\label{eq:p1_certifying_from_identity}
\begin{split}
p_1 (\ket{\psi}, \Omega_0) 
&= 1- \Tr\left(\Omega_0 \rho_0^{\ket{\psi}} \right)
= 1-\Tr\left(\left(\1 - \proj{\psi}\right) \rho_0^{\ket{\psi}} \right)
=\Tr\left(\proj{\psi} \rho_0^{\ket{\psi}} \right) \\
&= \bra{\psi} \left(\left(\Phi_0\otimes \1\right) 
(\proj{\psi}) \right)   \ket{\psi} <1,
\end{split}
\end{equation}
where the last inequality follows from the fact that $\Phi_0$ is not the 
identity channel.
Therefore, after sufficiently many queries in the parallel scheme the 
probability of false positive error will be arbitrarily small.
\end{proof}

Note that the expression for the probability of false positive error in 
Eq.~\eqref{eq:p1_certifying_from_identity} is in fact the fidelity between the 
input state and the output of the channel $\Phi_0$ extended by the identity 
channel. As we were not imposing any specific assumptions on the input state, 
we can take the one which minimizes the expression in 
Eq.~\eqref{eq:p1_certifying_from_identity}. Therefore,  the 
probability of the false positive error in the single-shot certification yields
\begin{equation}
p_1 = \min_{\ket{\psi}}
\bra{\psi} \left(\left(\Phi_0\otimes \1\right) 
(\proj{\psi}) \right)   \ket{\psi}.
\end{equation}
Eventually, to make sure that the probability of false positive error will not 
be greater than $\epsilon$, we will need $N_\epsilon \geq \left\lceil  
\frac{\log \epsilon}{\log p_1} \right\rceil$ steps in the parallel scheme.

From the above considerations we can draw a simple conclusion concerning the 
situation when $\Phi_0(X)= UX U^\dagger$ is a unitary channel. 
Then, as the unitary channel has only one Kraus operator, it holds that  $p_1 
=\min_{\ket{\psi}} \left\vert \bra{\psi} \left( U \otimes \1 \right) \ket{\psi} 
\right\vert^2 = \min_{\ket{\psi}} \left\vert \bra{\psi}  U  \ket{\psi} 
\right\vert^2 = \nu^2(U)$, where $\nu(U)$ is the distance from zero to the 
numerical range of the matrix 
$U$~\cite{puchala2021multiple,lewandowska2021optimal}. 
Thanks to this geometrical representation (see 
further~\cite{puchala2021multiple}) 
one can deduce the connection between the 
probability of false positive error, $p_1$,  and the probability of making an 
error in the unambiguous discrimination of unitary channels. More specifically, 
let $p^u_{\mathrm{error}}$ denote the probability of making an erroneous 
decision in unambiguous discrimination of unitary channels. Then, it holds that 
$p^u_{\mathrm{error}} = p_1^2$. Therefore, in the case of certification of 
unitary channels the probability of making the false positive error 
is significantly smaller than the probability of erroneous unambiguous 
discrimination.

\section{Certification of quantum 
measurements}\label{sec:certification_of_povms}

In this section we will take a closer look into the certification of quantum 
measurements. 
We will begin with general POVMs and later focus on the class of 
measurements with rank-one effects. 
Before stating the results, let us recall that every quantum measurement can be 
associated with quantum-classical channel defined as 
\begin{equation}
\PP(\rho) = \sum_i \tr (M_i \rho) \proj{i},
\end{equation}
where $\{M_i\}_i$ are measurement's effects and $\tr (M_i \rho)$ is the 
probability of obtaining the $i$-th label.

The following proposition can be seen as a corollary from 
Theorem~\ref{th:condition_for_parallel_certification} as it gives a simple 
condition when we forbid false negative error. This condition is expressed in 
terms of inclusion of supports of the measurements' effects.

\begin{proposition}\label{prop:certification_of_POVMs}
Let $\PP_0$ and $\PP_1$ be POVMs with effects $\{M_i\}_{i=1}^m$
and $\{N_i\}_{i=1}^m$ respectively. 
Then $\PP_0$ can be certified against $\PP_1$ in the parallel scheme
if and only if there exists a pair of effects $M_i$, $N_i$ for which 
$\supp(M_i) 
\not\subseteq \supp(N_i)$.
\end{proposition}

\begin{proof}
Let
\begin{equation}
M_i = \sum_{k_i} \alpha_{k_i}^i \proj{x_{k_i}^i}
\end{equation}
be the spectral decomposition of $M_i$ (where $\alpha_{k_i}^i>0$ for every 
$k$).
Then 
\begin{equation}
\begin{split}
\PP_0 (\rho) 
= \sum_i \proj{i} \tr(M_i \rho)
= \sum_i \sum_{k_i} \alpha_{k_i}^i \ketbra{i}{x_{k_i}^i} \rho 
\ketbra{x_{k_i}^i}{i}
\end{split}
\end{equation}
and hence the Kraus operators of $\PP_0$ are 
$\left\{ \sqrt{\alpha_{k_i}^i} \ketbra{i}{x_{k_i}^i} \right\}_{k_i,i}$.
Analogously, the Kraus operators of $\PP_1$ are
$\left\{ \sqrt{\beta_{k_i}^i} \ketbra{i}{y_{k_i}^i} \right\}_{k_i,i}$.

Therefore from Theorem~\ref{th:condition_for_parallel_certification} we have 
that $\PP_0$ can be certified against $\PP_1$ in the parallel scheme if and 
only if 
\begin{equation}
\SPAN \left\{ \sqrt{\alpha_{k_i}^i} \ketbra{i}{x_{k_i}^i} \right\}_{k_i,i} 
\not\subseteq
\SPAN \left\{ \sqrt{\beta_{k_i}^i} \ketbra{i}{y_{k_i}^i} \right\}_{k_i,i},
\end{equation}
that is when there exists a pair of effects $M_i$, $N_i$ for which $\supp(M_i) 
\not\subseteq \supp(N_i)$.
\end{proof}
The above proposition holds for any pair of quantum measurements.
In the case of POVMs with rank-one effects, the above condition can still be 
simplified to linear independence of vectors. This is stated as the following 
corollary.

\begin{corollary}\label{corr:certification_POVMs_with_rank_one_effects}
Let $\PP_0$ and $\PP_1$ be measurements with effects $\{ \alpha_i 
\proj{x_i}\}_{i=1}^m$ and $\{ \beta_i\proj{y_i}\}_{i=1}^m$  for $\alpha_i, 
\beta_i \in (0, 1]$, respectively.
Then $\PP_0$ can be certified against $\PP_1$ in the parallel scheme 
if and only if there exists a pair of vectors 
$\ket{x_i}$, $\ket{y_i}$  which are linearly independent.
\end{corollary}

While studying the certification of measurements with rank-one effects, one
cannot overlook their very important subclass, namely projective von Neumann
measurements. These measurements have effects of the form $\{ \proj{u_1}, \ldots
, \proj{u_n} \}$, where $\{\ket{u_i}\}_i$ form an orthonormal basis. This class
of measurements was studied in~\cite{lewandowska2021optimal}, though in a
slightly different context.  The main result of that work was the expression for
minimized probability of the false negative error, where the bound on the false
positive error was assumed. In this work, however, we consider the situation
when false negative error must be equal zero after sufficiently many uses.
Nevertheless, from
Corollary~\ref{corr:certification_POVMs_with_rank_one_effects} we can draw a
conclusion that 
any von Neumann measurement can be certified against some other von Neumann 
measurement
if and only if the measurements are not the same. 

\subsection{SIC POVMs}\label{sec:sic_povms_calculated}
Now we proceed to studying the certification of a special class of measurements
with rank-one effects, that is symmetric informationally complete (SIC) POVMs
\cite{renes2004symmetric,flammia2006sic,zhu2010sic,appleby2009properties}. We 
will directly calculate the bounds on the false
positive error in the single-shot and parallel certification. We will be using
the following notation. 
The SIC POVM $\PP_0$ with effects $\{\proj{x_i}\}_{i=1}^{d^2}$, where 
$\proj{x_i} = \frac{1}{d} \proj{\phi_i}$ and $\Vert \ket{\phi_i} \Vert = 1$, 
will be associated with the $H_0$ hypothesis.   
The SIC POVM $\PP_1$ corresponding to the alternative $H_1$ hypothesis will 
have effects $\{\proj{y_i}\}_{i=1}^{d^2}$, where 
$\proj{y_i} = \frac{1}{d} \proj{\phi_{\pi(i)}}$ and $\pi$ is a permutation of 
$d^2$ elements.
Moreover, the SIC condition assures that $| \braket{\phi_i}{\phi_{\pi(i)}}|^2 = 
\frac{1}{d+1}$ whenever $i \neq \pi(i)$.

\begin{remark}
From Corollary~\ref{corr:certification_POVMs_with_rank_one_effects} it follows
that for a SIC POVMs $\PP_0$ can be certified against SIC POVM $\PP_1$ in the  
parallel scheme as long as $\PP_0 \neq \PP_1$.
\end{remark}

Now we are working towards calculating the upper bound on the probability of  
the false positive error in single-shot certification of SIC POVMs. 
As the input state we take the maximally entangled state $\ket{\psi} \coloneqq 
\frac{1}{\sqrt{d}} \ketV{\1}$.
If the measurement was $\PP_0$, then the output state is
\begin{equation}
\rho_0^{\ket{\psi}} 
= \left(\PP_0 \otimes \1\right) \left( \proj{\psi}\right)
= \sum_{i=1}^{d^2} \proj{i} \otimes 
\frac{1}{d}(\proj{x_i})^\top
= \sum_{i=1}^{d^2} \proj{i} \otimes \frac{1}{d^2}(\proj{\phi_i})^\top,
\end{equation}
and similarly, if the measurement was $\PP_1$, then the output state is 
\begin{equation}
\rho_1^{\ket{\psi}} = \sum_{i=1}^{d^2} \proj{i} \otimes 
\frac{1}{d^2}(\proj{\phi_{\pi(i)}})^\top.
\end{equation}
As the output states have block-diagonal structure, 
we take the measurement effect to be in the block-diagonal form, that is
\begin{equation}\label{eq:Omega_SIC_construction_single_shot}
\Omega_0 \coloneqq \sum_{i=1}^{d^2} \proj{i} \otimes \Omega_i^\top,
\end{equation}
where for every $i$ we assume $\Omega_i \perp \proj{\phi_{\pi(i)}}$ to 
ensure that the probability of the false negative error is equal to zero.
We calculate 
\begin{equation}
\begin{split}
\tr \left(  \Omega_0 \rho_0 \right)
&= \tr \left(\left(  \sum_{i=1}^{d^2} \proj{i} \otimes \Omega_i^\top  \right)
\left(  \sum_{j=1}^{d^2} \proj{j} \otimes \frac{1}{d^2}(\proj{\phi_j})^\top  
\right)\right) \\
&= \tr \left( \sum_{i=1}^{d^2} \proj{i} \otimes \Omega_i^\top \frac{1}{d^2}
(\proj{\phi_i}) ^\top   \right) 
=\frac{1}{d^2} \sum_{i=1}^{d^2} \bra{\phi_i} \Omega_i \ket{\phi_i}.  
\end{split} 
\end{equation}

Let $k$ be the number of fixed points of the permutation $\pi$.
Taking $\Omega_i \coloneqq \1 - \proj{\phi_{\pi(i)}}$ we obtain

\begin{equation}
\begin{split}
\tr \left(  \Omega_0 \rho_0 \right)
&= \frac{1}{d^2} \sum_{i=1}^{d^2} \bra{\phi_i} \Omega_i \ket{\phi_i}
=  \frac{1}{d^2} \sum_{i=1}^{d^2} 
\bra{\phi_i} \left(\1 - \proj{\phi_{\pi(i)}}\right) \ket{\phi_i} \\
&= \frac{1}{d^2} \sum_{i=1}^{d^2} \left( 1- |\braket{\phi_i}{\phi_{\pi(i)}} 
|^2 
\right) 
= \frac{1}{d^2} \left(d^2 - k\right) \left( 1- |\braket{\phi_i}{\phi_{\pi(i)}} 
|^2 \right)  \\
&= \frac{1}{d^2} \left(d^2 - k\right) \left( 1- \frac{1}{d+1} \right)
= \frac{d^2 - k}{d^2 + d}.
\end{split}
\end{equation}
So far all the calculations were done for some fixed input state (maximally
entangled state) and measurement effect $\Omega_0$, which give us actually only
the upper bound on the probability of the false positive error. The current
choice of $\Omega_i = \1 - \proj{\phi_{\pi(i)}}$ seems like a good candidate,
but we do not know whether it is possible to find a better one. Using the
notation for the probability of the false positive error introduced in Eq.
\eqref{eq:def_of_p1} and \eqref{def_of_p1_conditional_single_shot}  we can write
our bound as
\begin{equation}
p_1 \leq p_1 \left(\ket{\psi}, \Omega_0 \right)= 1-\tr( \Omega_0 \rho_0) 
= \frac{d + k}{d^2+d}.
\end{equation}

On top of that, if $\pi$ does not have fixed points, that is when $k=0$, we have
$p_1 \leq \frac{1}{d+1}$ and the number of steps 
needed for  parallel certification is bounded by
$N_\epsilon \geq \left\lceil - \frac{\log \epsilon}{\log (d+1)} \right\rceil$.
In the case when the permutation $\pi$ has one fixed point, that is when $k=1$, 
it holds that $p_1 \leq \frac{1}{d}$ and hence the number of steps needed for 
parallel certification can be bounded by
$N_\epsilon \geq \left\lceil - \frac{\log \epsilon}{\log d} \right\rceil$.

\subsection{Parallel certification of SIC POVMs}\label{sec:sic_povms_parallel}
Let us consider a generalization of the results from previous subsection into 
the parallel scenario. We want to certify SIC POVMs $\PP_0$ and $\PP_1$ defined 
as in Subsection~\ref{sec:sic_povms_calculated}, however we assume that we are 
allowed to use the certified SIC POVM $N$ times in parallel. 
In this setup we associate the $H_0$ hypothesis with the measurement 
$\PP_0^{\otimes N}$, and analogously we associate  the $H_1$ hypothesis with 
the measurement  $\PP_1^{\otimes N}$.
It appears that the upper bound on false positive error is very similar to the
upper bound for the single-shot case. Straightforward but lengthy and technical
calculations give us
\begin{equation}\label{eq:sics_parallel_bound}
p_1^{\P, N}
 \leq \left( \frac{d + k}{d^2+d}\right)^N.
\end{equation}
The detailed derivation  of this bound is relegated to 
Appendix~\ref{app:sics_parallel}.

\section{Adaptive certification and Stein 
setting}\label{sec:adaptive_certification}
So far we were considering only the scheme in which the given channel is used 
a finite number of times in parallel. In this section we will focus on studying 
a more general scheme of certification, that is the adaptive certification.
In the adaptive scenario, we use the given channel $N$ times and between the 
uses we can perform some processing. It seems natural that the use of adaptive 
scheme instead of the simple parallel one should improve the certification. 
Surprisingly, in the case of von Neumann measurements the use of adaptive 
scheme gives no advantage over the parallel 
one~\cite{puchala2021multiple,lewandowska2021optimal}.
In other cases it appears that the use of processing is indeed a necessary 
step towards perfect 
discrimination~\cite{harrow2010adaptive,krawiec2020discrimination}.

Having the adaptive scheme as a generalization of the parallel one, let us take
a step further and take a look into the asymptotic setting. In other words, let
us discuss the situation when the number of uses of the certified channel tends
to infinity.  
There are various settings known in the literature concerning asymptotic
discrimination, like Stein and Hoeffding settings for asymmetric discrimination,
as well as Chernoff and Han-Kobayashi settings for symmetric discrimination. 
In the context of this work we will discuss only the setting concerning 
asymmetric discrimination, however
a concise introduction to all of these settings can be found e.g. 
in~\cite{wilde2020amortized}.
Arguably, the most well-known of these is the Hoeffding setting
which assumes the bound on the false negative error 
to be decreasing exponentially, and its area of interest is characterizing the 
error exponent of probability of false positive error.
Adaptive strategies for asymptotic discrimination in Hoeffding setting were 
recently explored in~\cite{salek2020adaptive}.
In the Stein  setting, on the other hand, we assume a constraint on the 
probability of false positive error and study the error exponent of the false 
negative error.  
Let us define a non-asymptotic quantity 
\begin{equation}
\zeta_n (\epsilon) \coloneqq
\sup_{\Omega_0, \ket{\psi} }
\left\{   
-\frac{1}{n} \log p_2^{\AA, n} \left( \ket{\psi}, \Omega_0 \right) : 
p_1^{\AA, n} \left( \ket{\psi}, \Omega_0 \right) \leq \epsilon
\right\},
\end{equation}
which describes the behavior of probabilities of errors in adaptive
discrimination scheme. The probability of false positive error after $n$ queries
is upper-bounded by some fixed $\epsilon$, and we are interested in studying how
quickly the probability of false negative error decreases. Therefore we consider
the logarithm of probability of false negative error divided by the number of
queries. Finally, a supremum is taken over all possible adaptive strategies,
that is we can choose the best input state, final measurement as well as the
processings between uses of the certified channel.

Note that in the previous sections we were considering $p_2^{\AA, N}$ instead 
of $p_2^{\AA, n}$, which in used in the Stein setting. The aim of this 
difference is to emphasize that in the Stein setting we study the situation in 
which the number of uses, $n$, tends to infinity. In contrary, in previous 
sections we were interested only in the case when the number of uses, $N$, was 
finite. 

Having introduced the non-asymptotic quantity $\zeta_n (\epsilon)$, let us
consider the case when the number of queries, $n$, tends to infinity. To do so,
we define the upper limit of the Stein exponent as
\begin{equation}\label{eq:stein_bounds}
\overline{\zeta} (\epsilon) \coloneqq
\limsup_{n \to \infty}
\zeta_n (\epsilon).
\end{equation}
Note that when $\overline{\zeta} (\epsilon)$ is finite, then  
the probability of the false negative error for adaptive certification will not
be equal to zero for any finite number of uses $N$. A very useful Remark $19$
from \cite{wilde2020amortized} states that $\overline{\zeta}(\epsilon)$ is
finite if and only if
\begin{equation}
\supp  \left( (\Phi_0 \otimes \1)(\proj{\psi_\text{ent}} ) \right)
\subseteq
\supp  \left( (\Phi_1 \otimes \1)(\proj{\psi_\text{ent}} ) \right), 
\end{equation}
where $\ket{\psi_\text{ent}}$ is the maximally entangled state.

Finally, we are in position to express the theorem stating the relation between 
adaptive and parallel certification.

\begin{theorem}\label{th:equivalence_parallel_adaptive}
Quantum channel $\Phi_0$ can be certified against quantum channel $\Phi_1$ in 
the parallel scenario if and only if quantum channel $\Phi_0$ can be certified 
against 
quantum channel $\Phi_1$ in the adaptive scenario.
\end{theorem}

Before presenting the proof of the Theorem we will state a useful lemma, which 
proof is postponed to Appendix~\ref{app:proofs_of_lemmas}.
\begin{lemma}\label{lm:stein}
Let $\overline{\zeta}(\epsilon)$ be as in 
Eq.~\eqref{eq:stein_bounds}. Then
$\overline{\zeta}(\epsilon)$ is finite if and only if
$\supp (\Phi_0) \subseteq \supp (\Phi_1)$.
\end{lemma}

\begin{proof}[Proof of Theorem \ref{th:equivalence_parallel_adaptive}]
When quantum channel $\Phi_0$ can be certified against the channel $\Phi_1$ in 
the parallel scenario, then naturally,
$\Phi_0$ can be certified against the channel $\Phi_1$ in 
the adaptive scenario.
Therefore it suffices to prove the reverse implication.  

Assume that the channel $\Phi_0$ can be certified against $\Phi_1$ in the 
adaptive scenario. This means that $\overline{\zeta}(\epsilon)$ is infinite. 
Hence from Lemma~\ref{lm:stein} it holds that $\supp (\Phi_0) \not\subseteq 
\supp (\Phi_1)$. Finally, from
Theorem~\ref{th:condition_for_parallel_certification}  we obtain that 
$\Phi_0$ can be certified against $\Phi_1$ in the 
parallel scheme.
\end{proof}

Theorem~\ref{th:equivalence_parallel_adaptive} states that 
if a quantum channel $\Phi_0$ can be certified against $\Phi_1$ in a finite 
number of queries, then the use of parallel scheme is always sufficient.
Therefore it may appear that adaptive
certification is of no value. Nevertheless, in some cases it still may be worth
using adaptive certification to reduce the number of uses of the certified
channel. For example in the case of SIC POVMs the use of adaptive scheme reduces
the number of steps significantly~\cite{krawiec2020discrimination}. A pair of
qutrit SIC POVMs can be discriminated perfectly after two queries in
adaptive scenario, therefore they can also be certified.
Nevertheless, they cannot be discriminated perfectly after any finite number of
queries in parallel.   
On the other hand, in the case of von Neumann measurements the number of steps
is the same no matter which scheme is used~\cite{puchala2021multiple}.

\section{Conclusions}\label{sec:conclusions}
As certification of quantum channels is in the NISQ era a task of significant
importance, the main aim of this work was to give an insight into this problem
from theoretical perspective. Certification was considered as an extension of
quantum hypothesis testing, which includes also preparation of an input state
and the final measurement.
We primarily focused on multiple-shot schemes of certification, that is our
areas of interest were mostly parallel and adaptive certification schemes. The
parallel scheme consists in certifying tensor products of channels while
adaptive scheme is the most general of all scenarios.

We derived a condition when after a finite number of queries in the parallel
scenario one can assure that the false negative error will not occur. We pointed
a class of channels which allow for excluding false negative error after a
finite number of uses in parallel but cannot be discriminated unambiguously. On
top of that, having a fixed upper bound on the probability of false positive
error, we found a bound on the number of queries needed to make the probability
of false positive error no greater than this fixed bound.

Moreover, we took into consideration the most general adaptive certification
scheme and studied whether it can improve the certification. It turned out that 
the use of parallel certification scheme is always sufficient to assure that the
false negative error will not occur after a finite number of queries.
Nevertheless, the number of queries needed to have the probability of false
positive error sufficiently small, may be decreased by using adaptive scheme.

\section*{Acknowledgments}

This work was supported by the Foundation for Polish Science (FNP) under grant
number POIR.04.04.00-00-17C1/18-00.
The project ,,Near-term quantum computers: Challenges, optimal
implementations and applications'' under Grant Number 
POIR.04.\\04.00-00-17C1/18-00,  is carried out within the Team-Net
programme of the Foundation for Polish Science co-financed by the
European Union under the European Regional Development Fund.

\section*{Author contributions}
All authors participated in  proving the theorems constituting the results 
of this work. {\L}.P. and Z.P. set the research objective and A.K. prepared the 
figures.

\section*{Competing interests}
The authors declare no competing interests.

\bibliographystyle{ieeetr}
\bibliography{certification}

\appendix
\section{Proofs of lemmas}\label{app:proofs_of_lemmas}

\begin{proof}[Proof of Lemma~\ref{lm:inclusion_of_supports}]
Suppose by contradiction that 
$\supp\left(\rho_0^{\ket{\psi}}\right) \subseteq 
\supp\left(\rho_1^{\ket{\psi}} , \ldots , \rho_m^{\ket{\psi}}\right)$, that is
\begin{equation}
\SPAN  \left\{  (E_i \otimes \1) \ket{\psi} \right\}_i
\subseteq
\SPAN  \left\{  \left(F^{(1)}_{j_1} \otimes \1\right) \ket{\psi} 
, \ldots ,
\left(F^{(m)}_{j_m} \otimes \1\right) \ket{\psi}
\right\}_{ {j_1}, \ldots,{j_m}}.
\end{equation}

Hence for every $i$
\begin{equation}
\begin{split}
(E_i \otimes \1) \ket{\psi} 
&= \sum_{j_1} \beta^{(1)}_{j_1} \left(F^{(1)}_{j_1} \otimes \1\right) 
\ket{\psi} 
+ \ldots +
\sum_{j_m} \beta^{(m)}_{j_m} \left(F^{(m)}_{j_m} \otimes \1\right) \ket{\psi} \\
&= \left(\sum_{j_1} \beta^{(1)}_{j_1} \left(F^{(1)}_{j_1} \otimes \1\right) 
+ \ldots +
\sum_{j_m} \beta^{(m)}_{j_m} \left(F^{(m)}_{j_m} \otimes \1\right) 
\right) \ket{\psi} \\
&= \left(\left(\sum_{j_1} \beta^{(1)}_{j_1} F^{(1)}_{j_1} 
+ \ldots +
\sum_{j_m} \beta^{(m)}_{j_m} F^{(m)}_{j_m} 
\right) \otimes \1 \right) \ket{\psi}, 
\end{split}
\end{equation}
where not all $\beta^{(k)}_{j_k}$ are  equal to zero.

As $\ket{\psi} \coloneqq \sum_t \lambda_t \ket{a_t} \ket{b_t}$, we have
\begin{equation}
(E_i \otimes \1) \ket{\psi}
= \sum_t \lambda_t \left( E_i \ket{a_t} \otimes\ket{b_t} \right)
\end{equation}
and
\begin{equation}
\begin{split}
&\left( \left(\sum_{j_1} \beta^{(1)}_{j_1} F^{(1)}_{j_1} 
+ \ldots +
\sum_{j_m} \beta^{(m)}_{j_m} F^{(m)}_{j_m} 
\right) \otimes \1 \right) \ket{\psi} \\
&= \sum_t \lambda_t
\left(\sum_{j_1} \beta^{(1)}_{j_1} F^{(1)}_{j_1} 
+ \ldots +
\sum_{j_m} \beta^{(m)}_{j_m} F^{(m)}_{j_m} 
\right) \ket{a_t} \otimes \ket{b_t}
\end{split}
\end{equation}

As $\{\ket{b_t}\}_t$ is an orthonormal basis, then for every $t$
\begin{equation}
E_i \ket{a_t} = 
\left(\sum_{j_1} \beta^{(1)}_{j_1} F^{(1)}_{j_1} 
+ \ldots +
\sum_{j_m} \beta^{(m)}_{j_m} F^{(m)}_{j_m} 
\right) \ket{a_t}, 
\end{equation}
and hence
\begin{equation}
E_i  = 
\sum_{j_1} \beta^{(1)}_{j_1} F^{(1)}_{j_1} 
+ \ldots +
\sum_{j_m} \beta^{(m)}_{j_m} F^{(m)}_{j_m}. 
\end{equation}
Therefore 
\begin{equation}
\SPAN \{E_i\}_i \subseteq \SPAN \left\{ F^{(1)}_{{j_1}} , \ldots ,  
F^{(m)}_{{j_m}} \right\}_{ {j_1}, \ldots,{j_m}},
\end{equation}
which implies that
$\supp(\Phi_0) \subseteq \supp \left(\Phi_1, \ldots , \Phi_m \right)$.

Finally, from the law of contraposition we obtain that if  
$\supp(\Phi_0) \not\subseteq \supp \left(\Phi_1, \ldots , \Phi_m \right)$, then 
$\supp\left(\rho_0^{\ket{\psi}}\right) \not\subseteq 
\supp\left(\rho_1^{\ket{\psi}} , \ldots , \rho_m^{\ket{\psi}}\right)$.
\end{proof}

\begin{proof}[Proof of Lemma~\ref{lm:inclusion_of_supports_2}]
Assume by contradiction that  $\supp (\Phi_0) \subseteq \supp \left(\Phi_1, 
\ldots \Phi_m\right)$, 
that is 
\begin{equation}
\SPAN \{E_i\}_i \subseteq \SPAN \left\{ F^{(1)}_{{j_1}} , \ldots ,  
F^{(m)}_{{j_m}} \right\}_{ {j_1}, \ldots,{j_m}}
\end{equation}
To simplify the notation, without loss of generality we define that
$\SPAN \left\{ F^{(1)}_{{j_1}} , \ldots ,  
F^{(m)}_{{j_m}} \right\}_{ {j_1}, \ldots,{j_m}} \eqqcolon \SPAN \{K_l\}_l$.
Hence for every natural number $N$ it also holds that 
\begin{equation}
\SPAN \left\{ E_{i_1} \otimes \ldots \otimes E_{i_N} \right\}_{i_1, \ldots , 
i_N}  
\subseteq 
\SPAN 
\left\{ K_{l_1} \otimes \ldots \otimes K_{l_N} \right\}_{l_1, \ldots , l_N}.
\end{equation}

Thus for every $i_1, \ldots , i_N$ we have that
\begin{equation}
E_{i_1} \otimes \ldots \otimes E_{i_N}
= \sum_{l_1, \ldots , l_N} \beta_{l_1, \ldots , l_N} 
K_{l_1} \otimes \ldots \otimes K_{l_N},
\end{equation}
where not all $\beta_{l_1, \ldots , l_N}$ are equal to zero.
Therefore for every $i_1, \ldots , i_N$ and input state $\ket{\psi}$, it 
also holds that
\begin{equation}
\begin{split}
\left( \left(E_{i_1} \otimes \ldots \otimes E_{i_N}\right) \otimes \1 
\right)\ket{\psi} 
&= \left( \left(\sum_{l_1, \ldots , l_N} \beta_{l_1, \ldots , l_N}  
K_{l_1} \otimes \ldots \otimes K_{l_N} 
\right) 
\otimes \1 \right)\ket{\psi} \\
&= \sum_{l_1, \ldots , l_N} \beta_{l_1, \ldots , l_N}  \left( 
\left(K_{l_1} \otimes \ldots \otimes K_{l_N} \right)
\otimes  
\1 
\right)\ket{\psi}
\end{split}
\end{equation}
and hence
\begin{equation}
\SPAN \left\{ \left( \left(E_{i_1} \otimes \ldots \otimes E_{i_N}\right) 
\otimes \1 
\right)\ket{\psi}\right\}_{i_1, \ldots , i_N}  
\subseteq
\SPAN \left\{ \left( \left(K_{l_1} \otimes \ldots \otimes K_{l_N} \right) 
\otimes \1 
\right)\ket{\psi}\right\}_{l_1, \ldots , l_N}.
\end{equation}
The above can be rewritten as 
\begin{equation}
\supp\left( \sigma_0^{N, \ket{\psi}}
\right) \subseteq
\supp \left( \sigma_1^{N, \ket{\psi}}, \ldots, \sigma_m^{N, \ket{\psi}}
\right), \quad N \in \mathbb{N}.
\end{equation}

Eventually, by the law of contraposition we obtain that if 
for some natural number $N$ and an input state $\ket{\psi}$ it holds that
\begin{equation}
\supp\left( \sigma_0^{N, \ket{\psi}}
\right) \not\subseteq
\supp \left( \sigma_1^{N, \ket{\psi}}, \ldots, \sigma_m^{N, \ket{\psi}}
\right),
\end{equation}
then $\supp (\Phi_0) \not\subseteq \supp \left(\Phi_1, \ldots \Phi_m\right)$.
\end{proof}

\begin{proof}[Proof of Lemma~\ref{lm:stein}]
($\implies$)
Let $\ket{\psi_\text{ent}}$ be the maximally entangled state.
When $\overline{\zeta}(\epsilon) < \infty$, then from 
Remark~$19$ from~\cite{wilde2020amortized} we have that
\begin{equation}
\supp  \left( (\Phi_0 \otimes \1)(\proj{\psi_\text{ent}} ) \right)
\subseteq
\supp  \left( (\Phi_1 \otimes \1)(\proj{\psi_\text{ent}} ) \right).
\end{equation}
From Lemma~\ref{lm:inclusion_of_supports}, (as the maximally entangled state 
has full Schmidt rank),  the above implies that 
$\supp (\Phi_0) \subseteq \supp (\Phi_1)$.

($\impliedby$) 
Now we assume that $\supp (\Phi_0) \subseteq \supp (\Phi_1)$.
From Lemma~\ref{lm:inclusion_of_supports_2}, this implies that for every 
natural number $N$ and every input state $\ket{\psi}$ it holds that
\begin{equation}
\supp  \left( (\Phi_0^{\otimes N} \otimes \1)(\proj{\psi} ) \right)
\subseteq
\supp  \left( (\Phi_1^{\otimes N} \otimes \1)(\proj{\psi} ) \right).
\end{equation}
Taking $N=1$ and $\ket{\psi} = \ket{\psi_\text{ent}}$ we obtain 
\begin{equation}
\supp  \left( (\Phi_0 \otimes \1)(\proj{\psi_\text{ent}} ) \right)
\subseteq
\supp  \left( (\Phi_1 \otimes \1)(\proj{\psi_\text{ent}} ) \right).
\end{equation}
Therefore from Remark~$19$ from~\cite{wilde2020amortized} we obtain that
$\overline{\zeta}(\epsilon) < \infty$.
\end{proof}

\section{Derivation of 
Eq.~\eqref{eq:sics_parallel_bound}}\label{app:sics_parallel}
Let $\PP_0$ and $\PP_1$ be as defined in 
Subsection~\ref{sec:sic_povms_calculated}. We consider the scenario where the 
certified measurement is used $N$ times in parallel. 
To calculate the bound on the parallel certification we will take particular 
choices of an input state and a final measurement. As for the input state, we 
will take the maximally entangled state, similarly as it was in the single-shot 
case. 
Applying tensor product of the SIC POVMs on the input state we obtain the 
output states either 
\begin{equation}
\begin{split}
\sigma_0^{N, \ket{\psi}} 
&= \frac{1}{d^{N}} \left( \PP_0 \otimes \ldots \otimes \PP_0 \otimes \1  
\right)
\left( \projV{\1}  \right) \\
&= \frac{1}{d^N} \sum_{i_1 , \ldots , i_N =1}^{d^2} \proj{i_1  \cdots i_N} 
\otimes \frac{1}{d^N}\left(\proj{\phi_{i_1} \cdots \phi_{i_N}}\right)^\top \\
&= \frac{1}{d^{2N}} \sum_{i_1 , \ldots , i_N =1}^{d^2} \proj{i_1  \cdots i_N} 
\otimes \left(\proj{\phi_{i_1} \cdots \phi_{i_N}}\right)^\top
\end{split}
\end{equation}
if the measurement was $\PP_0$, or
\begin{equation}
\begin{split}
\sigma_1^{N, \ket{\psi}} 
= \frac{1}{d^{2N}} \sum_{i_1 , \ldots , i_N =1}^{d^2} \proj{i_1  \cdots i_N} 
\otimes \left(\proj{\phi_{\pi(i_1)} \cdots \phi_{\pi(i_N)}}\right)^\top
\end{split}
\end{equation}
if the measurement was $\PP_1$. 
Similarly to the single-shot scenario, we take the measurement effect with 
block-diagonal structure, that is
\begin{equation}
\Omega_0 \coloneqq \sum_{i_1 , \ldots , i_N =1}^{d^2} \proj{i_1  \cdots i_N} 
\otimes 
\Omega_{i_1  \cdots i_N}^\top
\end{equation}
where we require $\Omega_{i_1  \cdots i_N} \perp \proj{\phi_{\pi(i_1)} \cdots
\phi_{\pi(i_N)}}$ to make sure that the false negative error will be equal zero.

We calculate
\begin{equation}
\begin{split}
\tr \left(\Omega_0 \sigma_0^{N, \ket{\psi}}  \right) 
&= \tr
\left(\left(  \sum_{i_1 , \ldots , i_N =1}^{d^2} \proj{i_1  \cdots i_N} \otimes 
\Omega_{i_1  \cdots i_N}^\top \right) \right.   \\
& \quad \quad \left. \left(\frac{1}{d^{2N}} \sum_{k_1 , \ldots , k_N =1}^{d^2} 
\proj{k_1 \cdots k_N} 
\otimes 
\left(\proj{\phi_{k_1} \cdots \phi_{k_N}}\right)^\top\right)
\right) \\
&= \frac{1}{d^{2N}} 
\sum_{i_1 , \ldots , i_N =1}^{d^2}
\sum_{k_1 , \ldots , k_N =1}^{d^2}
\tr 
\left(  \left(\proj{i_1  \cdots i_N} \otimes \Omega_{i_1  \cdots i_N}^\top 
\right) \right.  \\
& \left. \quad  \quad 
\left( \proj{k_1 \cdots k_N} \otimes 
\left(\proj{\phi_{k_1} \cdots \phi_{k_N}}\right)^\top\right)
\right) \\
&= \frac{1}{d^{2N}} 
\sum_{i_1 , \ldots , i_N =1}^{d^2}
\sum_{k_1 , \ldots , k_N =1}^{d^2}
\tr 
\left(\proj{i_1  \cdots i_N} \proj{k_1 \cdots k_N} 
\right.  \\
& \left. \quad  \quad 
\otimes \Omega_{i_1  \cdots i_N}^\top 
\left(\proj{\phi_{k_1} \cdots \phi_{k_N}}\right)^\top \right) \\
&= \frac{1}{d^{2N}} \sum_{i_1 , \ldots , i_N =1}^{d^2}  \tr 
 \left(\proj{i_1  \cdots i_N} \otimes \Omega_{i_1  \cdots i_N}^\top 
 \left(\proj{\phi_{i_1} \cdots \phi_{i_N}}\right)^\top \right)\\
&= \frac{1}{d^{2N}} \sum_{i_1 , \ldots , i_N =1}^{d^2}  \tr 
\left( \Omega_{i_1  \cdots i_N} \proj{\phi_{i_1} \cdots \phi_{i_N}} \right). 
\end{split}
\end{equation}
There are many possible choices of such $\Omega_{i_1  \cdots i_N}$ which 
fulfill the condition $\Omega_{i_1  \cdots i_N} \perp \proj{\phi_{\pi(i_1)} 
\cdots \phi_{\pi(i_N)}}$, but for the
time being we will take the one defined as follows
\begin{equation}
\Omega_{i_1  \cdots i_N} \coloneqq \1 - \proj{\phi_{\pi(i_1)} \cdots 
\phi_{\pi(i_N)}}.
\end{equation}
This choice of measurement effect may not appear optimal in general, but it is 
suitable for calculations due to its concise form. Therefore

\begin{equation}
\begin{split}
\tr \left(\Omega_0 \sigma_0^{N, \ket{\psi}}  \right) 
&= \frac{1}{d^{2N}} \sum_{i_1 , \ldots , i_N =1}^{d^2}  \tr 
\left( \left(  \1 - \proj{\phi_{\pi(i_1)} \cdots \phi_{\pi(i_N)}} \right) 
\proj{\phi_{i_1} \cdots \phi_{i_N}} \right) \\
&= \frac{1}{d^{2N}} \sum_{i_1 , \ldots , i_N =1}^{d^2}  
\left(  1-  
\left\vert  \braket{\phi_{i_1} \cdots \phi_{i_N}}{\phi_{\pi(i_1)} \cdots 
\phi_{\pi(i_N)}} 
\right\vert^2 \right) \\
&=1-  \frac{1}{d^{2N}} \sum_{i_1 , \ldots , i_N =1}^{d^2}  
\left\vert  \braket{\phi_{i_1} \cdots \phi_{i_N}}{\phi_{\pi(i_1)} \cdots 
\phi_{\pi(i_N)}} 
\right\vert^2.
\end{split}
\end{equation}
Therefore we have
\begin{equation}
p_1^{\P, N} \left(\ket{\psi}, \Omega_0 \right)  
= 1-\tr \left(\Omega_0 \sigma_0^{N, \ket{\psi}}  \right)
=\frac{1}{d^{2N}} \sum_{i_1 , \ldots , i_N =1}^{d^2}  
\left\vert  \braket{\phi_{i_1} \cdots \phi_{i_N}}{\phi_{\pi(i_1)} \cdots 
\phi_{\pi(i_N)}} 
\right\vert^2.
\end{equation}
To get the exact upper bound we need to calculate the sum, that is to explain 
that
\begin{equation}\label{eq:parallel_sics_combinatorics}
\sum_{i_1 , \ldots , i_N =1}^{d^2}  
\left\vert  \braket{\phi_{i_1} \cdots \phi_{i_N}}{\phi_{\pi(i_1)} \cdots 
\phi_{\pi(i_N)}} 
\right\vert^2
= \sum_{s=0}^{N} \binom{N}{N-s} k^{N-s} (d^2 - k)^s 
\frac{1}{(d+1)^s}.
\end{equation}
First, note that 
\begin{equation}
\left\vert  \braket{\phi_{i_1} \cdots \phi_{i_N}}{\phi_{\pi(i_1)} \cdots 
\phi_{\pi(i_N)}} \right\vert^2 
=\left\vert 
\braket{\phi_{i_1}}{\phi_{\pi(i_1)}} \cdots \braket{\phi_{i_N}}{\phi_{\pi(i_N)}}
 \right\vert^2
= \frac{1}{(d+1)^s}, 
\end{equation}
where $s \coloneqq | \{ i_l :\ i_l \neq \pi(i_l) \}|$.
In other words, every time we encounter a fixed point of the permutation we 
have a factor $\braket{\phi_{i_j}}{\phi_{\pi(i_j)}}$ which is equal one.
Let us now focus on consecutive factors of the right hand site of the 
Eq.~\eqref{eq:parallel_sics_combinatorics}.
The factor $\binom{N}{N-s}$ corresponds to choosing $N-s$ elements for which 
$\braket{\phi_{i_j}}{\phi_{\pi(i_j)}} = 1$. Then, on each of those elements 
there can be one of $k$ elements (as $k$ stands for the number of fixed points 
of the permutation $\pi$), therefore $k^{N-s}$. Then, on the remaining $s$ 
elements there can one of $d^2 - k$ values which are not fixed points of the 
permutation, hence we obtain $(d^2 - k)^s$. 
Further calculations reveal the concise expression for the upper bound on the 
probability of false negative error, that is 
\begin{equation}
p_1^{\P, N} \left(\ket{\psi}, \Omega_0 \right)  
= \frac{1}{d^{2N}} \sum_{s=0}^{N} \binom{N}{N-s} k^{N-s} (d^2 - k)^s 
\frac{1}{(d+1)^s}
=\left( \frac{d + k}{d^2+d}\right)^N.
\end{equation}
In the case of permutation $\pi$ having no fixed points, that is when 
$k=0$, the above bound simplifies to 
$p_1^{\P, N} \leq \left( \frac{1}{d+1}\right)^N$.

\end{document}